\documentclass[12pt]{article}
\usepackage{graphicx,amsmath,amssymb,url,enumerate,mathrsfs,epsfig,color}
\usepackage{tikz}
\usepackage{float}

\usepackage{amsmath}
\usepackage{amssymb}
\usepackage{yfonts}

\usepackage{tikz,pgfplots}

\usetikzlibrary{calc, patterns,arrows, shapes.geometric}

\usepackage{graphicx,amsmath,amssymb,url,enumerate,mathrsfs,epsfig,color}

\usetikzlibrary{decorations.text}
\usetikzlibrary{decorations.markings}

\pgfplotsset{compat=1.8}

\usepackage{enumerate}
\usepackage{boxedminipage}
\usepackage{makeidx}
\usepackage{multicol}
\usepackage{xcolor}
\usepackage{mathtools}

\usetikzlibrary{calc, patterns,arrows, shapes.geometric}

\newenvironment{proof}{\noindent{\bf Proof~:}}{\QED\medskip}
\def\QED{\hskip0.1em\hfill\null\ \null\nobreak\hfill
\kern3pt\lower1.8pt\vbox{\hrule\hbox
{\vrule\kern1pt\vbox{\kern1.7pt \hbox{$\scriptstyle
QED$}\kern0.2pt}\kern1pt\vrule}\hrule}}

\newcommand{\hor}{\mathrlap{\perp}\square}
\newcommand{\ram}{\raisebox{0.22ex}{$-$}}
\newcommand{\ver}{{\mathrlap{\ram}}\square}

\newtheorem{prop}{Proposition}[section]

\newtheorem{example}{Example}[section]
\title{Double groupoids in the theory of material uniformity}
\author{Marcelo Epstein \\ \small{University of Calgary, Canada} \\ \small{mepstein@ucalgary.ca}}
\begin{document}
\maketitle

\begin{abstract}
The use of double groupoids and their associated double Lie algebroids and characteristic distributions is proposed for the description and analysis of continuous media that carry two different constitutive or geometric structures. Various measures of misalignment and lack of uniformity that may arise when two individually perfectly uniform structures are combined are suggested, and the concept of coupled symmetries is proposed and illustrated.
\end{abstract}

\section{Introduction}

In continuum physics, the theory of material uniformity addresses the question of comparing the properties of each pair of points in a material substrate, assumed to be continuous. Uniformity is declared when the result of this investigation is that, according to a preestablished criterion, those properties are identical for every pair. Even if the substrate turns out to be uniform, however, the question  can be raised as to whether there exists a configuration in which all points are simultaneously in the same state. A positive answer to this question gives rise to the notion of homogeneity. Lack of homogeneity of a uniform body can be related to various physical theories of continuous distributions of dislocations and other defects.

In the general description provided above, the subtlety of a precise specification of what the material properties being compared may be has been deliberatley avoided. It is possible, therefore, to contemplate the coexistence of different kinds of properties in the same material substrate. An obvious example in this direction is the theory of material composites. In a binary composite, for example, two solid media are blended to form a mixture. Each of these components may be perfectly uniform in its own right, but the result will not be so in general. To visualize this possibility, one may think of two thin wooden plates that have been glued together after having given to one of them an in-plane deformation (induced perhaps by an arbitrary temperature field before the gluing takes place).

The composite, however, is not the only possibility where two uniformity structures coexist. One may think of an ordinary body in which, under microscopic observation, a geometrical lattice can be distinguished. This lattice may have little or no bearing on the mechanical properties of the substrate, but it may be decisive in the determination of other properties, such as optical or magnetic. Such a lattice, in the continuous limit, provides us with a frame field embedded in the substrate, characterized mathematically as a distant parallelism in the body. The mechanical properties themselves are also known to determine such a parallelism (or a family thereof), in a manner advocated most forcefully by Noll \cite{noll}. In this way, we obtain an interplay between two differential geometric structures, of different physical origins, cohabiting within the same medium. Is the body simultaneoously uniform with respect to both structures?

For obvious reasons, as even this cursory introduction has suggested, the language in which the various theories of uniformity and homogeneity have been traditionally framed is that of differential geometry. Increasingly sophisticated tools of this discipline have been invoked to this end. The line of thought stemming from \cite{noll} and further developed in a similar vein and described in works such as \cite{wang, bloom, ees, eelz}, has been successfully unified under the umbrella of the theory of groupoids \cite{mms, book}. In the case of the coexistence of two different strucures, the individuality of each of which wants to be preserved in the mix, a more refined tool appears to be required. In a recent article \cite{symmetry}, it has been argued that this tool is to be found in the theory of double groupoids. The purpose of the present work is to further this argument and to develop new results along the same lines of enquiry.

To furnish a common starting point, in Section \ref{sec:measures}, the uniformity of a binary composite is tackled from a more standard viewpoint. For components with various kinds of symmetries, a measure of non-uniformity of the composite is suggested. These preliminary and rather intuitive results are later, in Section \ref{sec:algebroid}, placed within the setting of the double Lie algebroid associated with the composite. Section \ref{sec:layered} discusses the possible non-uniform structures that may arise in the case of a non-uniform composite made of two uniform components. These include layered and filamented structures. Section \ref{sec:groupoid} contains a general introduction to the concepts of groupoids and double groupoids and their applications to binary structures. Finally, Section \ref{sec:algebroid} broaches the infinitesimal version of a double groupoid, namely the associated double Lie algebroid and the associated characteristic distribution. These notions are discussed at a somewhat intuitive level that attempts to avoid some of the traps of excessive formalism.

\section{Measures of non-uniformity} \label{sec:measures}

\subsection{Review of some basic notions}

For the present purposes, a {\it binary composite} is defined as a non-reactive mixture of two elastic solids, or a combination of two parallelisms, as described in the introduction. In this section we will consider the case in which both components are individually smoothly materially uniform, and we will determine necessary and sufficient conditions for the mixture to be uniform. 

Assuming smoothness, material uniformity can be interpreted in terms of a family of smooth frame fields ${\bf P}(X)$, where $X$ is a variable point in the body manifold $\mathcal B$. Each such frame field determines a linear connection $\Gamma_P$. Restricting the frame field to a sufficiently small open neighbourhood ${\mathcal U} \subset {\mathcal B}$, the connection $\Gamma_P$ is equivalent to a distant parallelism, that is, a connection with vansihing curvature in $\mathcal U$, but with possibly non-vanishing torsion. When the symmetry group of the material is discrete, it can be shown that the material connection is unique.

A frame ${\bf P}(X)$ at a point can be regarded as an invertible linear map ${\bf P}(X): {\mathbb R}^3 \to T_X{\mathcal B}$ to the tangent space $T_X{\mathcal B}$, as schematically depicted in Figure \ref{fig:implants}. In a uniform elastic body $\mathcal B$, this map can be regarded as an {\it implant} of an {\it archetype}, which is usually thought of as a material point in a standard reference state (such as a natural configuration). The composition ${\bf P}_{XY}={\bf P}(Y) \; {\bf P}^{-1}(X)$ is an invertible linear map from $T_X{\mathcal B}$ to $T_Y{\mathcal B}$ known as a {\it material isomorphism from $X$ to $Y$}. Its physical meaning is a transplant that achieves a perfecgt graft, so that the material response of the grafted material element after being deformed according to ${\bf P}_{XY}$ is mechanically indistinguishable from the original element at $Y$. Clearly, this objective can be achieved only if both points are made of the same material.

 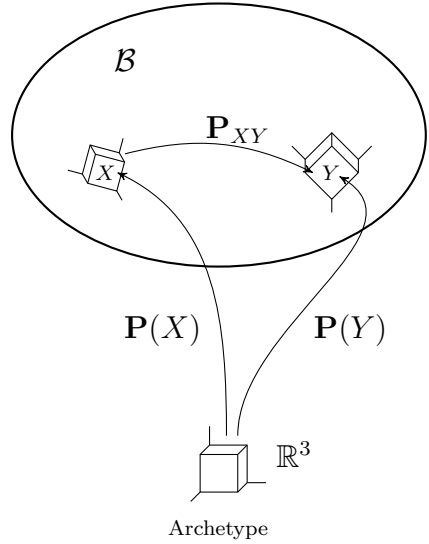
\begin{figure}[H]
\begin{center}

\begin{tikzpicture}[scale = 0.5]
\begin{scope}[shift={(0,-3)}, scale=0.5]
\draw[black] (1,1) -- (1,-1) -- (-1,-1) -- (-1,1) -- (1,1); 
\draw[black] (-1,1) -- (-0.5,1.5) -- (1.5,1.5) -- (1.5,-0.5) --
(1,-1); \draw[black] (1,1) -- (1.5,1.5);
\draw (1.5,-0.5)--(2.5,-0.5);
\draw (-1,-1)--(-1.5,-1.5);
\draw (-0.5,1.5)--(-0.5,2.5);
\end{scope}
\begin{scope}[shift={(3,5)},rotate=45, scale=0.5]
\draw[black] (1,1) -- (1,-1) -- (-1,-1) -- (-1,1) -- (1,1); 
\draw[black] (-1,1) -- (-0.5,1.5) -- (1.5,1.5) -- (1.5,-0.5) --
(1,-1); \draw[black] (1,1) -- (1.5,1.5);
\draw (1.5,-0.5)--(2.5,-0.5);
\draw (-1,-1)--(-1.5,-1.5);
\draw (-0.5,1.5)--(-0.5,2.5);
\end{scope}
\begin{scope}[shift={(-3,5)},rotate=75,scale=0.4]
\draw[black] (1,1) -- (1,-1) -- (-1,-1) -- (-1,1) -- (1,1); 
\draw[black] (-1,1) -- (-0.5,1.5) -- (1.5,1.5) -- (1.5,-0.5) --
(1,-1); \draw[black] (1,1) -- (1.5,1.5);
\draw (1.5,-0.5)--(2.5,-0.5);
\draw (-1,-1)--(-1.5,-1.5);
\draw (-0.5,1.5)--(-0.5,2.5);
\end{scope}

\draw[thick] (0,6) ellipse (5.5cm and 3.5cm); \draw[-stealth'] (0.5,-2) to [out=90,in=-30]
(3.2,4.9); \draw[-stealth'] (0.2,-2) to [out=90,in=-25]
(-2.7,5.0);\node at(3.5,0.8) {${\bf P}({Y})$}; \node
at(0,-4.5) {$_{\rm Archetype}$};
\node at (-1.5,0.8) {${\bf P}(X)$};
\node at (3,5) {$_Y$};
\node at (-3,5) {$_X$};
\draw[-stealth'] (-2.5,5.5) to [bend left=20] (2.5,5);
\node[above] at (0.5,5.5) {${{\bf P}_{XY}}$};
\node at (2,-2.5) {${\mathbb R}^3$};
\node at (-2.5,8) {$\mathcal B$};

\end{tikzpicture}
\end{center}
\caption{Implants and material isomorphism}
\label{fig:implants}
\end{figure}

The Christoffel symbols of the connection $\Gamma_P$ associated with the smooth frame in a neighbourhood ${\mathcal U } \subset {\mathcal B}$ are obtained as
\begin{equation} 
\Gamma^I_{JK}= P^I_\alpha\;P^{\alpha}_{J,K}= -P^I_{\alpha,K}\;P^{-\alpha}_J.
\end{equation}
In these expressions we have used Greek indices to denote components in the archetype (that is, in the standard basis of ${\mathbb R}^3$) and capital indices to refer to components in a coordinate system (usually Cartesian) in the reference configuration of the body. Moreover, commas denote partial differentiation with respect to the coordinates $X^I$, and a minus sign preceding a superscript indicates components of the inverse map. The summation convention is used throughout. The torsion of this connection is the skew-symmetric tensor $\tau$ with components
\begin{equation}
\tau^I_{JK}=\Gamma^I_{JK}-\Gamma^I_{KJ}.
\end{equation}

For solid materials, by definition, the symmetry group in any local reference configuration must be a conjugate of a sugbroup of the orthogonal group in ${\mathbb R}^3$. It follows, therefore, that, quite apart from the material connections just described, there is also a metric connection associated with the Riemannian metric induced by one (and, therefore, all) of the frame fields. This unique metric connection is, clearly, torsion-free, but may have a non-vanishing Riemannian curvature. The metric tensor $\bf g$ of this Riemannian metric is defined as ${\bf g}=({\bf PP}^T)^{-1}$. To avoid any misunderstanding of the meaning of the transposed ${\bf P}^T$, we can express the metric tensor $\bf g$ in components as
\begin{equation}
g_{IJ}= \delta_{\alpha\beta}\;P^{-\alpha}_I \;P^{-\beta}_J.
\end{equation}

\subsection{Suggested measures}

For a binary composite, the brief conceptual framework just introduced can be applied independently to each of the two components, which we will indicate with  subscripts 1 and 2. We will next suggest conditions for the composite to be materially uniform when the individual frame fields ${\bf P}_1$ and ${\bf P}_2$ (at the moment of effecting the mixture, say) are known.

\begin{enumerate}
\item\label{case1} {\it Two components with discrete symmetry groups}: The composite is uniform if, and only if, the (unique) material connections coincide. Thus, we can define the measure of local non-uniformity of the composite by means of the tensor
\begin{equation} \label{eq1}
{\bf B}= \Gamma_1 - \Gamma_2.
\end{equation}
 If $P^I_\alpha$ denotes the components of the uniformity frame of the first component, the non-uniformity tensor $\bf B$ can be expressed as
\begin{equation} \label{eq6}
B^I_{JK}=- P^{-\alpha}_{\;J} P^I_{\alpha;K}
\end{equation}
where a semicolon indicates the covariant derivative with respect to the material connection of the second component.
\item\label{case2} {\it A discretely symmetric and a fully isotropic component}:  A necessary and sufficient condition for uniformity of this composite is the coincidence of the material metrics of the components. The measure of local non-uniformity can be defined in this case as
\begin{equation} \label{eq3}
{\bf B}= {\bf g}_1-{\bf g}_2.
\end{equation}
\item {\it A discretely symmetric and a transversely isotropic component}: A frame field of a transversely isotropic solid implies the existence of a vector field ${\bf n}(X)$ in the body. A necessary and sufficient condition for uniformity of the composite is the equality of the material metrics and the parallelism of $\bf n$ with respect to the discrete connection $\Gamma_1$. In other words, the measure of non-uniformity of the composite is given by the two tensors
\begin{equation}
{\bf B}= {\bf g}_1-{\bf g}_2,\;\;\;\;\;\;\;\;\;\;{\hat{\bf B}}=\nabla_1{\bf n},
\end{equation}
where $\nabla_1$ denotes the covariant derivative with respect to the unique material connection $\Gamma_1$.
\item {\it Two fully isotropic components}: A necessary and sufficient condition for uniformity is the coincidence of the Riemannian metrics. The measure of local non-uniformity is
\begin{equation}
{\bf B}= {\bf g}_1-{\bf g}_2,
\end{equation}
just as in Case \ref{case2} above. 
\item {\it Two transversely isotropic components}: A necessary and sufficient condition for uniformity is the coincidence of the Riemannian metrics, and the angle condition
\begin{equation}
\langle {\bf n}_1, {\bf P}_1 {\bf P}_2^{-1}{\bf n}_2\rangle_M=\langle {\bf n}_1, {\bf n}_2 \rangle_M.
\end{equation}
In this expression, $\langle\;,\;\rangle_M$ denotes the inner product relative to the common Riemannian metric.
 
\end{enumerate}

The symmetry group at a point $X$ of the composite is equal to the intersection of the symmetry groups of the components at that point. Thus, the symmetry group of a composite made of two fully isotropic materials can be as small as the trivial group or as large as the orthogonal group.

\subsection{Remarks on time evolution}

A uniform binary composite, just as any other uniform material, may undergo processes of material evolution, such as remodeling, aging, and morphogenesis. What is peculiar to binary composites is that, even if the material is not only uniform but also homogeneous, the partial stresses in each of the components, may not vanish in a homogeneous configuration. This fact may be a further trigger of evolutionary processes. If the composite is not even uniform, the non-uniformity itself may become a driving force behind the evolution. In particular contexts, the lack of uniformity may be interpreted, for example, as some kind of difference of chemical potential that tends to reduce the measure of non-uniformity. In a different context, the composite may be seen as an ordinary material that carries an embedded network of fibers not necessarily contributing significantly to the material properties but associated with other optical or biological phenomena. In such cases, the non-uniformity of the composite may become the driving force for a time evolution of this extra structure within an unchanging material background.

A first-order evolution law for a material frame field expresses the time derivative of the map ${\bf P}(X,t)$ as a function of a list of variables that include $\bf P$ itself and other physically meaningful quantities. An evolution driven exclusively by the non-uniformity of a composite will, therefore, include the non-uniformity tensor and, possibly, its spatial gradients. If one of the components is non-evolving and if, moreover, it is in a homogeneous reference configuration, we recover the results of \cite{eelz, e1} provided that the dependence on $\bf B$ is assumed to be mediated in Cases \ref{case1} and \ref{case2}, respectively, by the torsion of the variable material connection or the curvature of its material metric.

\section{Layered structures}
\label{sec:layered}

The uniformity of the components does not guarantee the uniformity of the composite, as we have demonstrated. It may happen, however, that certain regions of the composite are uniform. In the most general case, as we know from \cite{jgp1}, a non-uniform body can be regarded as the union of materially uniform subsets. In the smooth case, we obtain a material foliation in which each leaf is a body submanifold of dimension 0, 1, 2, or 3. From a physical standpoint, these submanifolds are isolated points, uniform fibers, uniform laminae, or uniform sub-bodies, respectively. In the case of binary composites, this classification can be expressed in terms of the non-uniformity tensors introduced in Section \ref{sec:measures}.

To illustrate these ideas, let us consider a composite of type \ref{case1}, whose non-uniformity tensor is given by Equation (\ref{eq1}). The action ${\bf Bv}$ of the third-order tensor $\bf B$ on a vector $\bf v$ defined in components by $B^I_{JK} v^K$ is, therefore, a measure of the change of the components of $\bf P$ along $\bf v$ as seen from the perspective of the second frame. We can define the null space of $\bf B$ at a point $X \in {\mathcal B}$ as the vector subspace of $T_X{\mathcal B}$ spanned by all vectors $\bf v$ such that ${\bf Bv=0}$.\footnote{For a thorough treatment of eigenvalues and eigenvectors of third-order tensors see \cite{qi}.}

If the dimension $m$ of the null space of $\bf B$ is constant over $\mathcal B$ we obtain a regular distribution. For $m=0$ we recognize a totally non-uniform body, such as a functionally graded material in which  there are no uniform submanifolds except isolated points. The case $m=1$ corresponds to a fibered body, in which each fiber is uniform. From the mathematical point of view, the smoothness of the material parallelisms ${\bf P}_1$ and ${\bf P}_2$ guarantees, via the theorem of existence and uniqueness of solutions of systems of ordinary differential equations, that the fibers are properly defined.

The case $m=2$ gives rise to a regular two-dimensional distribution on $\mathcal B$. According to the theorem of Frobenius, this distribution is completely integrable (that is, it has integral embedded manifolds of dimension $m$) if, and only if, the distribution is involutive. This theorem raises the question as to whether a body with null spaces of $\bf B$ of constant dimension 2 is actually a laminated body, since the corresponding distribution may turn out not to be involutive. The following proposition answers this question.

\begin{prop} A regular distribution generated by the null spaces of $\bf B$ is always involutive.
\end{prop}
\begin{proof} Let $\bf v$ and $\bf w$ be everywhere linearly independent vector fields that belong everywhere to the null space of $\bf B$, namely
\begin{equation} \label{eq7}
B^I_{JK}v^K=B^I_{JK}w^K=0.
\end{equation}
We need to show that the Lie bracket $[{\bf v}, {\bf w}]$ belongs to the distribution. We have
\begin{equation} \label{eq8}
B^I_{JK}[{\bf v}, {\bf w}]^K = B^I_{JK}\left(v^Lw^K_{,L}-w^Lv^K_{,L}\right) =-B^I_{JK,L}(v^L w^K-w^Lv^K),
\end{equation}
where Equation (\ref{eq7}) has been invoked repeatedly. We note that, for any frame $\bf P$, the partial derivatives of the connection symbols can be written as
\begin{equation} \label{eq9}
\Gamma^I_{JK,L} = P^I_\alpha P^{-\alpha}_{\;\;\;J,KL}-\Gamma^M_{JK} \Gamma^I_{ML}.
\end{equation}
Denoting, for clarity, the connections $\Gamma_1$ and $\Gamma_2$ in Equation (\ref{eq1}) as $\Gamma$ and $\hat \Gamma$, respectively, and using Equations (\ref{eq8}) and (\ref{eq9}), we can write
\begin{equation} \label{eq10}
\begin{split}
B^I_{JK}[{\bf v}, {\bf w}]^K & = -B^I_{JK,L}(v^L w^K-w^Lv^K) \\
& = \left(-\Gamma^M_{JK} \Gamma^I_{ML}+{\hat\Gamma}^M_{JK} {\hat\Gamma}^I_{ML}\right)(v^L w^K-w^Lv^K),
\end{split}
\end{equation}
in which the symmetry of the second partial derivatives in (\ref{eq9}) has been exploited to cancel some terms. By definition, ${\bf B}=\Gamma-{\hat \Gamma}$, which yields
\begin{equation} \label{eq11}
-\Gamma^M_{JK} \Gamma^I_{ML}+{\hat\Gamma}^M_{JK} {\hat\Gamma}^I_{ML}=-B^M_{JK} B^I_{ML} - B^M_{JK} {\hat \Gamma}^I_{ML}-{\hat\Gamma}^M_{JK} B^I_{ML}.
\end{equation}
Plugging this result into Equation (\ref{eq10}) and invoking (\ref{eq7}), completes the proof.
\end{proof}

\section{The groupoid perspective} \label{sec:groupoid}
\subsection{Preliminaries}

In continuum mechanics, a {\it material symmetry} is a symmetry of a constitutive law under the action of a group. In particular applications, this broad definition is further restricted to the appropriate context. For a {\it simple elastic material}, the constitutive equation of a body $\mathcal B$ in a given reference configuration has the form of a function
\begin{equation}
{\bf T}={\bf T}({\bf F}, X).
\end{equation}
In this equation, $\bf T$ stands for the Cauchy stress tensor, and $\bf F$ is the deformation gradient, both tensors being evaluated at the point $X \in {\mathcal B}$. A material symmetry at $X$ is a symmetry of this function under the right action of the special linear group $SL(3;{\mathbb R})$. More specifically, a symmetry in this case is represented by a matrix $\bf G$ of unit determinant, such that
\begin{equation}
{\bf T}({\bf F G}, X)={\bf T}({\bf F}, X)
\end{equation}   
is satisfied identically for all ${\bf F} \in GL(3;{\mathbb R})$, that is, for all non-singular $3 \times 3$ matrices $\bf F$. The {\it symmetry group} ${\mathcal G}_X$ at $X \in {\mathcal B}$ is the collection of all these symmetries under the operation of matrix multiplication.

The symmetries just described may be regarded as {\it local symmetries}, in the sense that they are associated with a particular body point. There is, however, a different sense in which the word symmetry may be understood. Indeed, if two different points $X, Y \in {\mathcal B}$ happen to be made of the same material, we have a case of a {\it distant symmetry} with an obvious physical meaning. As described in Section \ref{sec:measures}, we may think of it as a {\it transplant operation}, by means of which we mentally cut a small neighbourhood of $X$ and transplant it, after a suitable deformation ${\bf P}_{XY}$, as a substitute of the original neighbourhood of $Y$. If this operation results in a mechanically indistiguishable response at $Y$, we can assert that the two points are made of the same material. This kind of symmetry results in the equation
\begin{equation}
{\bf T}({\bf F},Y)={\bf T}({\bf FP}_{XY},X)
\end{equation}
to be satisfied identically for all non-singular $\bf F$. Under this light, a local symmetry is nothing but a particular case of a distant symmetry when $Y=X$. A distant material symmetry is also known as a {\it material isomorphism} between $X$ and $Y$, a terminology introduced by Walter Noll \cite{noll}.

Do all distant symmetries of a body $\mathcal B$ form a group? In a group, by definition, any two elements can be composed, giving rise to a new element of the group. It is clear that not all distant symmetries can be composed. Indeed, if ${\bf P}_{XY}$ is a material isomorphism between $X$ and $Y$, and ${\bf P}_{WZ}$ is a material isomorphism between $W$ and $Z$, then the composition ${\bf P}_{WZ}{\bf P}_{XY}$ makes no sense unless $Y=W$. This impasse can be solved with the use of the concept of {\it groupoid}, a mathematical structure that generalizes the notion of group in several respects.\footnote{The notion of groupoid was introduced by Brandt in 1927 \cite{brandt}. An excellent motivational presentation is Weinstein's article \cite{weinstein}.}  Intuitively, a groupoid can be regarded as a collection of arrows between points of a base set $\mathcal B$. The arrows are composable if, and only if, they are connected in a tip-to-tail fashion. Moreover, for every arrow $u$ between two points $X$ and $Y$, there is a well-defined inverse arrow $u^{-1}$ from $Y$ to $X$. The compositions $uu^{-1}$ and $u^{-1}u$ give rise to loop-shaped arrows at $Y$ and $X$, known as the {\it units} at the respective points. Endowing the body with the groupoid structure generated by the material isomorphisms, one speaks of the {\it material groupoid}.

Consider next the case of a composite  material, a solid mixture, or a material with several independent constitutive structures. One obvious way to deal with these cases is to identify the shared symmetries between the components, and to obtain what can be regarded as the intersection of the individual material groupoids. A more general point of view, however, consists in attributing an official status, as it were, to some clever combination of the various structures as a point of departure for the analysis. Thus, in the case of two components, can we combine the corresponding two material groupoids into a single superstructure that encompasses both? 

The positive answer to this question can be obtained by placing the physical implications of shared symmetries within the formal framework of the so-called {\it double groupoids}.\footnote{Double categories and double groupoids were introduced in 1963 by Charles Ehresmann \cite{ehresmann}.} The basic building block of a double groupoid can be thought of as the cooperative result of two pairs of material isomorphisms that have joined forces to achieve a common goal, as vaguely suggested in Figure \ref{fig:hands}.

 \begin{figure}[H]
\begin{center}
\begin{tikzpicture} [scale=1.0]
\draw  (2.1,0.35)--(0.2,0.3)--(0.2,0.8) -- (2,0.8) to[out=60,in=205] (2.5,1.1)--(2.5,1)--(3.2,1) to [out=0, in=0] (3.2,0.8) -- (2.85,0.8) -- (3.2,0.8) to [out=0, in=0] (3.2,0.6)--(2.8,0.6)--(3.15,0.6) to [out=0, in=0] (3.15,0.4)--(2.78,0.4)--(3.08,0.4) to [out=0, in=0] (3.08,0.25)--(2.6,0.25) to [out=180,in=-45] (2.1,0.35) ;


\begin{scope} [rotate around={270:(1.80,1.5)},yshift=0]
\draw[fill=gray!50]  (2.1,0.35)--(0.2,0.3)--(0.2,0.8) -- (2,0.8) to[out=60,in=205] (2.5,1.1)--(2.5,1)--(3.2,1) to [out=0, in=0] (3.2,0.8) -- (2.85,0.8) -- (3.2,0.8) to [out=0, in=0] (3.2,0.6)--(2.8,0.6)--(3.15,0.6) to [out=0, in=0] (3.15,0.4)--(2.78,0.4)--(3.08,0.4) to [out=0, in=0] (3.08,0.25)--(2.6,0.25) to [out=180,in=-45] (2.1,0.35) ;
\end{scope}
\begin{scope} [rotate around={180:(1.80,1.5)},yshift=0]
\draw[fill=white]  (2.1,0.35)--(0.2,0.3)--(0.2,0.8) -- (2,0.8) to[out=60,in=205] (2.5,1.1)--(2.5,1)--(3.2,1) to [out=0, in=0] (3.2,0.8) -- (2.85,0.8) -- (3.2,0.8) to [out=0, in=0] (3.2,0.6)--(2.8,0.6)--(3.15,0.6) to [out=0, in=0] (3.15,0.4)--(2.78,0.4)--(3.08,0.4) to [out=0, in=0] (3.08,0.25)--(2.6,0.25) to [out=180,in=-45] (2.1,0.35) ;
\end{scope}
\begin{scope} [rotate around={90:(1.80,1.5)},yshift=0]
\draw[fill=gray!50]  (2.1,0.35)--(0.2,0.3)--(0.2,0.8) -- (2,0.8) to[out=60,in=205] (2.5,1.1)--(2.5,1)--(3.2,1) to [out=0, in=0] (3.2,0.8) -- (2.85,0.8) -- (3.2,0.8) to [out=0, in=0] (3.2,0.6)--(2.8,0.6)--(3.15,0.6) to [out=0, in=0] (3.15,0.4)--(2.78,0.4)--(3.08,0.4) to [out=0, in=0] (3.08,0.25)--(2.6,0.25) to [out=180,in=-45] (2.1,0.35) ;
\end{scope}
\draw[fill=white]  (2,0.8) to[out=60,in=205] (2.5,1.1)--(2.5,1)--(3.2,1) to [out=0, in=0] (3.2,0.8) -- (2.85,0.8) -- (3.2,0.8) to [out=0, in=0] (3.2,0.6)--(2.8,0.6)--(3.15,0.6) to [out=0, in=0] (3.15,0.4)--(2.78,0.4)--(3.08,0.4) to [out=0, in=0] (3.08,0.25)--(2.6,0.25) to [out=180,in=-45] (2.1,0.35) ;
\end{tikzpicture}
\end{center}
\caption{Artistic rendition of the building block of a double groupoid}
\label{fig:hands}
\end{figure}
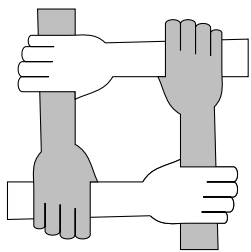

\subsection{Groups, groupoids, double groupoids}

Recall that a group is a set $\mathcal G$ endowed with a binary internal operation (called {\it product}, {\it multiplication}, or {\it composition}). This operation, indicated by apposition of two elements, is associative, namely, $(ab)c=a(bc)$, for all $a,b,c \in {\mathcal G}$. Moreover, there is a unique {\it unit} element $e \in {\mathcal G}$ such that $eg=ge=g$, for all $g \in {\mathcal G}$. Finally, for each element $g \in {\mathcal G}$, there is a unique {\it inverse} element $g^{-1}$, such that $g^{-1}g=gg^{-1}=e$. The uniqueness of the unit and the inverses does not need to be assumed ab initio, since it can be proved from the definitions.

In a groupoid, a binary associative operation is defined in a set $\mathcal Z$, just as in the case of a group, except that the operation is not generally defined for all pairs of elements. Associativity means that if the products $uv$ and $vw$ are defined, then the products $(uv)w$ and $u(vw)$ are defined and are equal. Moreover, just as in a group, for every $u \in {\mathcal Z}$ there is a unique inverse $u^{-1} \in {\mathcal Z}$ such that both products $u^{-1}u$ and $uu^{-1}$ are defined. At this point, however, a crucial difference with the concept of group makes its appearance, as these two products function as units, but are not necessarily equal. More precisely, if the product $uv$ is defined, then $uvv^{-1}=u$ and $u^{-1}uv=v$.

This compact and elegant definition can be replaced with an equivalent one that, although longer, is more amenable to applications.\footnote{The link between the two definitions is provided by the selection of a distinguished subset of $\mathcal Z$ consisting of all elements expressible in the form $uu^{-1}$ for some $u \in {\mathcal Z}$. This subset is then identified with the set of objects, which is the base set of the groupoid.} According to this definition, a groupoid ${\mathcal Z} \rightrightarrows {\mathcal B}$ consists of two sets, namely, a {\it total set} $\mathcal Z$ (or {\it set of arrows}) and a {\it base set} $\mathcal B$ (or {\it set of objects}). These two sets are linked by two surjective projection maps, $\alpha:{\mathcal Z} \to {\mathcal B}$ and $\beta:{\mathcal Z} \to {\mathcal B}$, called the {\it source map} and the {\it target map}, respectively. The definition, so far, conjures up the picture of a collection of arrows whose tails and tips are objects. The maps $\alpha$ and $\beta$ assign to each arrow its tail (or source) and its tip (or target), respectively. Not every pair of points is necessarily connected by an arrow. But if an arrow $u \in {\mathcal Z}$ has source $X=\alpha(u)$ and target $Y=\beta(u)$, then there is also an {\it inverse} arrow $u^{-1} \in {\mathcal Z}$ such that $Y=\alpha(u^{-1})$ and $X=\beta(u^{-1})$. A binary associative operation (product, or composition) $uv$ is defined in $\mathcal Z$ between two arrows $u, v \in {\mathcal Z}$ if, and only if, $\alpha(u)=\beta(v)$, that is, when the tip of $v$ coincides with the tail of $u$, as shown in Figure \ref{fig:inverse}.  The arrow $uv$ has the tail of $v$ and the tip of $u$. Just as in the composition of functions, the second factor is applied first. Finally, to each object $X \in {\mathcal B}$ a {\it unit} arrow $\epsilon(X) \in {\mathcal Z}$ is assigned such that $u^{-1}u=\epsilon(\alpha(u))$ and $uu^{-1}=\epsilon(\beta(u))$, and such that $u=u\epsilon(\alpha(u))=\epsilon(\beta(u))u$, for all $u \in {\mathcal Z}$. A groupoid is {\it transitive} if for each pair of objects, $X$ and $Y$, there is at least one arrow $u$ such that $\alpha(u)=X$ and $\beta(u)=Y$.

 \begin{figure}[H]
\begin{center}
\begin{tikzpicture} [scale=0.6]
\tikzset{->-/.style={decoration={
  markings,
  mark=at position .65 with {\arrow{stealth'}}},postaction={decorate}}}
  
  \draw[thick, o-o,->-] (3,0) -- (0,0); 
  \draw [thick,o-, ->-] (6,0) -- (3,0);
  \draw[thick,->-](5.85,-0.05) to [bend left] (0.15,-0.05);
  \node[left] at (0,0) {$Z$};
  \node [above] at (3,0) {$Y$};
  \node[right] at (6,0) {$X$};
  \node at (1.5,.7) {$u$};
    \node at (4.5,.7) {$v$};
      \node at (3,-1.3) {$uv$};
\node at (3,-2.5) {(a) Composition};
\begin{scope}[xshift=12cm]
  \draw[thick, ->-] (0,0) to [bend left] (6,0); 

  \draw[thick,->-](6,0) to [bend left] (0,0);
  \node[above] at (0.25,0.1) {$X$};

  \node[above] at (6,0.1) {$Y$};
  \node at (3,1.3) {$u$};

      \node at (3,-1.3) {$u^{-1}$};

      \draw[thick, ->-] (0,0) arc (360:5:0.5);
          \draw[thick, ->-] (6,0) arc (180:-175:0.5);
          \node[below] at (-1.45,-0.15) {$\epsilon(X)$};
            \node[below] at (7.4,-0.15) {$\epsilon(Y)$};
\node at (3,-2.5) {(b) Inverse and units};
\end{scope}
\end{tikzpicture}
\end{center}
\caption{Groupoid operations}
\label{fig:inverse}
\end{figure}
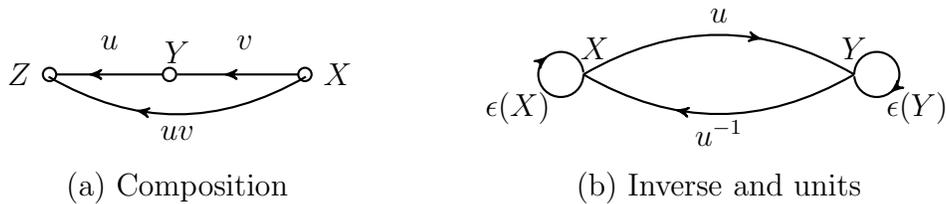

Starting from the general definition of a groupoid, it is not difficult to check that the collection of all loop-shaped arrows at a point $X$, namely the collection of all arrows $u$ such that $\alpha(u)=\beta(u)=X$, constitutes a group, known as the {\it vertex group at $X$}. Moreover, one can easily prove that in a transitive groupoid all the vertex groups are mutually conjugate. In a material groupoid, where the arrows are identified with material isomorphisms, transitivity is tantamount to uniformity. Furthermore, the vertex groups are precisely the symmetry groups of the material at each point.

In the definition of a groupoid, the nature of the sets $\mathcal Z$ and $\mathcal B$ is left unspecified. If, for example, these sets happen to be smooth manifolds, and if the maps involved have certain desirable properties, we obtain the category of Lie groupoids. As an example, a smoothly uniform material body gives rise to a Lie groupoid. Leaving these smoothness considerations aside, it may happen that a groupoid has a base set that is itself a groupoid, whose base is another set.\footnote{This situation is not uncommon in other mathematical structures. A good example is the iterated tangent bundle of a differentiable manifold.} This idea would lead to the concept of 2-groupoid. A double groupoid is a somewhat more involved concept, whose definition will be presently reviewed.

\begin{example}\label{ex:pair} {\rm{\bf The pair groupoid}: For any set of objects $\mathcal B$, the Cartesian product ${\mathcal B}^2={\mathcal B} \times {\mathcal B}$ has a natural groupoid structure, called {\it the pair groupoid over $\mathcal B$}. Each arrow can be represented as an ordered pair $u=(Y,X)$, such that $\alpha(u)=Y$ and $\beta(u)=X$.\footnote{The only reason to adopt the second projection as the source, and the first projection as the target, is one of notational elegance. If $v=(X,Y)$ and $u=(Z,Y)$, then $uv=(Z,Y)(Y,X)=(Z,X)$.} The unit at $X \in {\mathcal B}$ is $\epsilon(X)=(X,X)$, and the inverse of $u=(X,Y)$ is $u^{-1}=(Y,X)$.The pair groupoid is transitive and all its vertex groups are trivial. If $\mathcal B$ is a smooth manifold, the pair groupoid is a Lie groupoid.}
\end{example}

A {\it double groupoid} consists of a set $\mathcal Q$ endowed with two groupoid structures, ${\mathcal Q} \rightrightarrows {\mathcal V}$ and ${\mathcal Q} \rightrightarrows {\mathcal H}$, called, respectively, the {\it horizontal} and the {\it vertical} structures. Each of the base spaces is itself a groupoid over a common base $\mathcal B$, namely ${\mathcal H}\rightrightarrows {\mathcal B}$ and ${\mathcal V}\rightrightarrows {\mathcal B}$, called, respectively, the horizontal and vertical {\it side groupoids}. A common notation for a double groupoid is
\begin{equation} \label{dg1}
\begin{matrix}
{\mathcal Q}                   &\rightrightarrows  & {\mathcal V} \\
  \downdownarrows    &                      &  \downdownarrows  \\
  {\mathcal H}                 & \rightrightarrows   &\mathcal{B}
\end{matrix}
\end{equation}
Since there are 4 different groupoids involved in the definition, it is necessary to establish a notation for the elements, maps, and operations of each. Authors differ in this respect, each author adopting a notation that seems more convenient to particular applications. The notations adopted herein are in part a mixture of the notations in \cite{ehresmann}, \cite{brown1}, and \cite{natale}. 

The elements of $\mathcal B$ will be denoted with capital letters $(X, Y, ...)$. The arrows, units, and projections of the side groupoid ${\mathcal H}\rightrightarrows {\mathcal B}$  will follow the usual notation for groupoids introduced above. For the vertical side groupoid  ${\mathcal V}\rightrightarrows {\mathcal B}$, the same notation will be used, except that a hat (circumflex accent) will be added. The products will be indicated by simple apposition, since there is no room for confusion. Thus, for example, $uv$ is an arrow in $\mathcal H$, and ${\hat p} {\hat q}$ is an arrow in $\mathcal V$. The projections and units will be also distinguished by the presence or absence of a hat, to avoid any possible ambiguity. Thus, $\alpha(u)$ is the source of $u \in {\mathcal H}$ and ${\hat \alpha}({\hat p})$ is the source of ${\hat p} \in {\mathcal V}$ .

The elements of $\mathcal Q$ will be denoted with double-strike letters $({\mathbb A},{\mathbb B}, ...)$, and represented graphically as squares (much in the same way as the elements of $\mathcal H$ and $\mathcal V$ are represented as arrows). A square has two horizontal parallel sides belonging to $\mathcal H$, while the vertical sides belong to $\mathcal V$. Moreover, there are two different products between squares, one for each of the two groupoid structures  ${\mathcal Q}\rightrightarrows {\mathcal V}$ and ${\mathcal Q}\rightrightarrows {\mathcal H}$, denoted respectively by $\hor$ and $\ver$.

The projections and the units in each of the two groupoid structures of $\mathcal Q$ will be signaled with a superimposed tilde and with a subscript $_H$ or $_V$ for the horizontal and vertical structures, respectively. Note that the source and target of the horizontal structure belong to $\mathcal V$, while their couterparts for the vertical structure belong to $\mathcal H$. To exemplify this notation, we draw an element ${\mathbb A} \in {\mathcal Q}$  in the shape of a square, as anticipated, namely,
\begin{equation} \label{dg2}
\begin{tikzpicture}[baseline=(current  bounding  box.center)]
  \draw[thick,-stealth'] (3,0) -- (0,0);
  \draw[thick,-stealth'] (3,3) -- (0,3);
  \draw[thick,-stealth'] (3,0) -- (3,3);
  \draw[thick,-stealth'] (0,0) -- (0,3);
\node at (1.5,1.5) {$\mathbb A$};
\node at (4.2,1.5) {${\hat s}={\tilde \alpha}_H({\mathbb A})$};
\node at (-1.3,1.5) {${\hat t}={\tilde \beta}_H({\mathbb A})$};
\node[above] at (1.5,3) {$t={\tilde \beta}_V({\mathbb A})$};
\node[below] at (1.5,0) {$s={\tilde \alpha}_V({\mathbb A})$};
\end{tikzpicture}
\end{equation}

In the horizontal structure, the product ${\mathbb A}\hor{\mathbb A}'$ can only be carried out if ${\hat s}={\hat t}'$, in agreement with the usual concatenation criterion. Graphically,
\begin{equation} \label{dg3}
\begin{tikzpicture}[baseline=(current  bounding  box.center)]
\foreach \x in {0,1.5,4.5}
{  \draw[thick,-stealth'] (1.5+\x,0) -- (0+\x,0);
  \draw[thick,-stealth'] (1.5+\x,1.5) -- (0+\x,1.5);
  \draw[thick,-stealth'] (1.5+\x,0) -- (1.5+\x,1.5);
  \draw[thick,-stealth'] (0+\x,0) -- (0+\x,1.5);}
\node[left] at(-0.,0.75) {${\mathbb A}\hor {\mathbb A}'\;=\;{\hat t}$};
\node[above] at(0.75,1.5) {$  t$};
\node[above] at(2.25,1.5) {$ t'$};
\node[below] at(0.75,-0.1) {$  s$};
\node[below] at(2.25,0) {$ s'$};
\node[right] at (3,0.75) {${\hat s}'\;=\;{\hat t}$};
\node[above] at (5.25,1.5) {$tt'$};
\node[below] at (5.25,0) {$ss'$};
\node[right] at (6,0.75) {${\hat s}'$};
\end{tikzpicture}
\end{equation}
The unit square of the horizontal structure at the vertical arrow $\hat u$ is the square
\begin{equation} \label{dg3a}
\begin{tikzpicture}[baseline=(current  bounding  box.center)]
\foreach \x in {0}
{  \draw[thick,-stealth'] (1.5+\x,0) -- (0+\x,0);
  \draw[thick,-stealth'] (1.5+\x,1.5) -- (0+\x,1.5);
  \draw[thick,-stealth'] (1.5+\x,0) -- (1.5+\x,1.5);
  \draw[thick,-stealth'] (0+\x,0) -- (0+\x,1.5);}
\node[left] at(-0.,0.75) {${\tilde\epsilon}_H({\hat u})\;=\;{\hat u}$};
\node[above] at(0.75,1.5) {$ \epsilon({\hat \beta}({\hat u}))$};
\node[below] at(0.75,-0.1) {$ \epsilon({\hat \alpha}({\hat u}))$};
\node[right] at (1.5,0.75) {${\hat u}$};
\end{tikzpicture}
\end{equation}

Similarly, in the vertical structure, the product ${\mathbb A}\ver{\mathbb A}'$ can only be carried out if $s'=t'$, or graphically
\begin{equation} \label{dg4}
\begin{tikzpicture}[baseline=(current  bounding  box.center)]
\foreach \x in {0,-1.5}
{  \draw[thick,-stealth'] (1.5,0+\x) -- (0,0+\x);
  \draw[thick,-stealth'] (1.5,1.5+\x) -- (0,1.5+\x);
  \draw[thick,-stealth'] (1.5,0+\x) -- (1.5,1.5+\x);
  \draw[thick,-stealth'] (0,0+\x) -- (0,1.5+\x);}
\node[left] at(-0.,0) {${\mathbb A}\ver {\mathbb A}'\;=\;\;\;\;\;$};
\node[left] at (0,0.75) {${\hat t}$};
\node[above] at(0.75,1.5) {$  t$};
\node[right] at (1.5,0.75) {${\hat s}$};
\node[right] at (1.5,-0.75) {${\hat s}'$};
\node[below] at(0.75,-1.5) {$s'$};
\node[left] at (0,-0.75) {$\hat t'$};
\node[right] at (1.8,0.1) {$\;\;=\;\;\;{\hat t} {\hat t}'$};
\foreach \x in {3.5}
{\draw[thick,-stealth'] (1.5+\x,-0.75) -- (0+\x,0-0.75);
  \draw[thick,-stealth'] (1.5+\x,1.5-0.75) -- (0+\x,1.5-0.75);
  \draw[thick,-stealth'] (1.5+\x,0-0.75) -- (1.5+\x,1.5-0.75);
  \draw[thick,-stealth'] (0+\x,0-0.75) -- (0+\x,1.5-0.75);}
\node[right] at (5,0.1) {${\hat s}{\hat s}'$};
\node[above] at (4.25,0.75) {$t$};
\node[below] at (4.25,-0.75) {$s'$};
\end{tikzpicture}
\end{equation}
The unit square of the vertical structure at the horizontal arrow $ u$ is the square
\begin{equation} \label{dg4a}
\begin{tikzpicture}[baseline=(current  bounding  box.center)]
\foreach \x in {0}
{  \draw[thick,-stealth'] (1.5+\x,0) -- (0+\x,0);
  \draw[thick,-stealth'] (1.5+\x,1.5) -- (0+\x,1.5);
  \draw[thick,-stealth'] (1.5+\x,0) -- (1.5+\x,1.5);
  \draw[thick,-stealth'] (0+\x,0) -- (0+\x,1.5);}
\node[left] at(-0.,0.75) {${\tilde\epsilon}_V({u})\;=\;{ {\hat\epsilon}({ \beta}({ u}))}$};
\node[above] at(0.75,1.5) {$ u$};
\node[below] at(0.75,-0.1) {$ u$};
\node[right] at (1.5,0.75) {${ {\hat\epsilon}({ \alpha}({ u}))}$};
\end{tikzpicture}
\end{equation}

The definition of a double groupoid must be completed by the imposition of a compatibility condition between the two structures, namely,
\begin{equation} \label{dg5}
({\mathbb A}\hor{\mathbb B})\ver({\mathbb C}\hor{\mathbb D})=({\mathbb A}\ver{\mathbb C})\hor({\mathbb B}\ver{\mathbb D}),
\end{equation}
whenever the operations are possible. Graphically, this condition means that, in a square made of four squares whose edges in contact match, the result is the same whether one first composes horizontally and then vertically, or vice versa. In other words, the large square shown below makes sense.
\begin{equation} \label{dg6}
\begin{tikzpicture}[baseline=(current  bounding  box.center)]
\foreach \x in {0,1.5}
{\foreach \y in {0,1.5}
{  \draw[thick,-stealth'] (1.5+\x,0+\y) -- (0+\x,0+\y);
  \draw[thick,-stealth'] (1.5+\x,1.5+\y) -- (0+\x,1.5+\y);
  \draw[thick,-stealth'] (1.5+\x,0+\y) -- (1.5+\x,1.5+\y);
  \draw[thick,-stealth'] (0+\x,0+\y) -- (0+\x,1.5+\y);}}
\node at (0.75,0.75) {$\mathbb C$};
\node at (0.75,2.25) {$\mathbb A$};
\node at (2.25,0.75) {$\mathbb D$};
\node at (2.25,2.25) {$\mathbb B$};
\end{tikzpicture}
\end{equation}

Finally, the {\it double source map} assigning to each square $\mathbb A$ its sources in the horizontal and vertical structures, ${\tilde \alpha}_H({\mathbb A}), {\tilde \alpha}_V({\mathbb A})$, is often assumed to be surjective. The graphical representation of this assumption is that a right angle made of a horizontal and a vertical arrow issuing from a common corner in $\mathcal B$, can always be completed (not necessarily uniquely) to a square in $\mathcal Q$, as suggested below. This is known as the {\it filling condition}.
\begin{equation} \label{dg7}
\begin{tikzpicture}[baseline=(current  bounding  box.center)]
\draw[thick,-stealth'] (0,0)--(-2,0);
\draw[thick,-stealth'] (0,0)--(0,2);
\draw[thick,-stealth'] (4,0)--(2,0);
\draw[thick,-stealth'] (4,0)--(4,2);
\draw[thick,-stealth'] (4,2)--(2,2);
\draw[thick,-stealth'] (2,0)--(2,2);
\node at (1,1) {$\Longrightarrow$};
\end{tikzpicture}
\end{equation}

A useful example of a double groupoid is the {\it coarse double groupoid}, $\square({\mathcal H},{\mathcal V})$, generated by two groupoids, $\mathcal H$ and $\mathcal V$, with a common base set. It consists of all the possible consistent squares that can be formed using the two groupoids as sides. Any double groupoid can be naturally mapped as an inclusion into the coarse double groupoid generated by its side groupoids. A {\it double Lie groupoid} is obtained when the four groupoids involved in the algebraic definition of a double groupoid are Lie groupoids, with the additional condition that the double source map is a surjective submersion.

Given a double groupoid $\mathcal Q$, we define the {\it transposed double groupoid} ${\mathcal Q}^T$ by exchanging the side groupoids. Accordingly, each square ${\mathbb A}\in\mathcal Q$ gives rise to a correspoding transposed square ${\mathbb A}^T \in {\mathcal Q}^T$, as shown in Figure \ref{fig:transpose}. Notice the exchange of the position of the corners $X$ and $Y$.
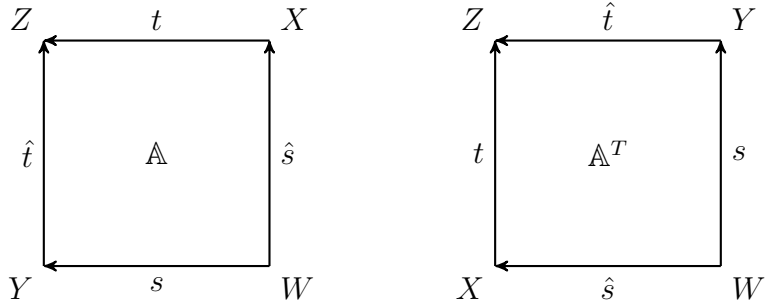
\begin{figure} [H]
\begin{center}
\begin{tikzpicture}[baseline=(current  bounding  box.center)]
  \draw[thick,-stealth'] (3,0) -- (0,0);
  \draw[thick,-stealth'] (3,3) -- (0,3);
  \draw[thick,-stealth'] (3,0) -- (3,3);
  \draw[thick,-stealth'] (0,0) -- (0,3);
\node at (1.5,1.5) {$\mathbb A$};
\node[right] at (3,1.5) {${\hat s}$};
\node[left] at (0,1.5) {${\hat t}$};
\node[above] at (1.5,3) {$t$};
\node[below] at (1.5,0) {$s$};
\node[below right] at ((3,0) {$W$};
\node[above right] at ((3,3) {$X$};
\node[above left] at ((0,3) {$Z$};
\node[below left] at ((0,0) {$Y$};
\begin{scope} [xshift=6cm]
  \draw[thick,-stealth'] (3,0) -- (0,0);
  \draw[thick,-stealth'] (3,3) -- (0,3);
  \draw[thick,-stealth'] (3,0) -- (3,3);
  \draw[thick,-stealth'] (0,0) -- (0,3);
\node at (1.5,1.5) {${\mathbb A}^T$};
\node[right] at (3,1.5) {${s}$};
\node[left] at (0,1.5) {${t}$};
\node[above] at (1.5,3) {$\hat t$};
\node[below] at (1.5,0) {$\hat s$};
\node[below right] at ((3,0) {$W$};
\node[above right] at ((3,3) {$Y$};
\node[above left] at ((0,3) {$Z$};
\node[below left] at ((0,0) {$X$};
\end{scope}
\end{tikzpicture}
\end{center}
\caption{A square and its transposed}
\label{fig:transpose}
\end{figure}

The {\it core} of a double groupoid consists of the collection of all squares for which $W=X=Y$ and in which the arrows $s$ and ${\hat s}$ are, respectively, the unit arrows of the horizontal and vertical groupoids at $W$. The core can be given a canonical groupoid structure by regarding its arrows as pairs of arrows $(t,{\hat t})$ with common source $W$ and target point $Z$. In many respects, the core groupoid characterizes the double groupoid \cite{brown1}, although the physical identity of the original side groupoids is not preserved when passing to the core.

\begin{example}\label{ex:pairdouble} {\rm{\bf The pair double groupoid}: The Cartesian product ${\mathcal B}^4={\mathcal B}\times{\mathcal B}\times {\mathcal B} \times{\mathcal B}$ has a double groupoid natural structure, which can be identified with the coarse groupoid $\square({\mathcal B}^2, {\mathcal B}^2)$. The two side groupoids are equal to the pair groupoid introduced in Example \ref{ex:pair}. In reference to the diagram on the left of Figure \ref{fig:transpose}, we can write: ${\tilde \alpha}_V({\mathbb A})=(Y,W)$, ${\tilde \beta}_V({\mathbb A})=(Z,X)$, ${\tilde \alpha}_H({\mathbb A})=(X,W)$, and ${\tilde \beta}_H({\mathbb A})=(Z,Y)$}
\end{example}

\subsection{The material double groupoid of a composite}

Each of the two components of a binary composite gives rise to a well-defined material groupoid, ${\mathcal Z}\rightrightarrows {\mathcal B}$ and ${\hat {\mathcal Z}}\rightrightarrows {\mathcal B}$. The arrows of each of these groupoids represent material isomorphisms, which implies that they are both subgroupoids of $\Pi^1(\mathcal B, \mathcal B)$, the 1-jets groupoid of $\mathcal B$, whose arrows are the linear isomorphisms between the tangent spaces, $T_X{\mathcal B}$ and $T_Y{\mathcal B}$, of all pairs of points $X, Y \in {\mathcal B}$. The {\it material double groupoid} of a binary composite has been first defined in \cite{symmetry} as the double groupoid $\mathcal Q$ whose side groupoids are $\mathcal Z$ and ${\hat{\mathcal Z}}$, and whose squares satisfy the {\it commutation condition}
\begin{equation} \label{double1}
t{\hat s} = {\hat t}s,
\end{equation}
using the notation of Equation (\ref{dg2}).

The commutation condition (\ref{double1}) makes sense precisely because both side groupoids of $\mathcal Q$ are subgroupoids of the same groupoid $\Pi^1(\mathcal B, \mathcal B)$. Nevertheless, even when both side groupoids are transitive, there is no reason to expect that any non-trivial commutative squares exist. A legitimate question, therefore, is the investigation of the physical meaning of the commutative squares in the double groupoid of a binary composite.

\subsection{A measure of misalignment}

For the sake of clarity, let us consider a case in which both components are uniform (i.e., both side groupoids are transitive) and triclinic (i.e., their structure groups are trivial). As a consequence of these assumptions, given any pair of points $X, Y \in {\mathcal B}$, there is always a unique arrow ($u$ and $u^*$, say) of each side groupoid with source $X$ and target $Y$, as shown in Figure \ref{fig:misalign1}. The {\it misalignment} of $Y$ with respect to $X$ is defined as the composition $m={u^*}^{-1} u$. The particular case $m=\epsilon(X)$, corresponds to the two points {\it in the composite} being materially isomorphic. Otherwise, there is an actual deviation from isomorphism. This definition can be regarded as a discrete version of the tensor $\bf B$ defined in Equation (\ref{eq1}).
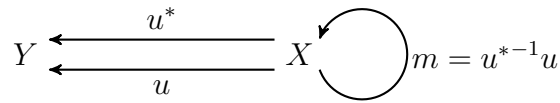
\begin{figure}[H]
\begin{center}
\begin{tikzpicture}[baseline=(current  bounding  box.center)]
  \draw[thick,-stealth'] (3,0) -- (0,0);
  \draw[thick,-stealth'] (3,0.4) -- (0,0.4);
\node[below] at (1.5,0) {$u$};
\node[above] at (1.5,0.4) {$u^*$};
\draw[thick,-stealth'] (3.6,-0.) arc(-160:160:0.6);
\node[right] at (3,0.2) {$X$};
\node[left] at (0,0.2) {$Y$};
\node at (5.8,0.2) {$m={u^*}^{-1} u$};
\end{tikzpicture}
\end{center}
\caption{Misalignment}
\label{fig:misalign1}
\end{figure}

While $m$ is certainly an arrow in $\Pi(\mathcal B, \mathcal B)$, it fails in general to be an arrow in either of the side groupoids. Under a change of reference configuration, the misalignment arrow is affected only by the local  deformation at $X$. Indeed, let $\phi: {\mathcal B} \to {\mathbb R}^3$ be the change of configuration, then the misalignment $\phi(m)$ of $\phi(Y)$ with respect to $\phi(X)$ is given by the composition
\begin{equation} \label{double2}
\phi(m)= H\circ m \circ  H^{-1},
\end{equation}
where $H=j^1_X \phi$, that is, the derivative of $\phi$ at $X$.

Given two pairs of points, $X,Y$ and $X',Y'$, with misalignments $m$ and $m'$ respectively, a fair comparison cannot be invariantly defined due to the aforementioned dependence on the local reference configuration. Let us, however, effect a change of configuration $\phi_1$ that brings both {\it first components} of the composite to the same local rereference at points $X$ and $X'$. If the misalignments $\phi_1(m)$ and $\phi_1(m')$ coincide (that is, if they differ by the effect of a mere Euclidean translation), we say that the pairs $X, Y$ and $X', Y'$ are {\it 1-compatible}. Physically speaking, these two pairs have the same degree of non-uniformity. This property can be regarded as a special kind of {\it coupled symmetry} between pairs of points.

An equivalent, more concise, definition of 1-compatibility between the pairs $X,Y$ and $X', Y'$ is obtained from the right-hand side of Equation (\ref{double2}) by interpreting $H$ as the material isomomorphism between the first component at $X$ and $X'$. In other words, the pairs are 1-compatible if, and only if, $m'= H\circ m \circ  H^{-1}$.

In a similar fashion, we can define the notion of {\it 2-compatibility}, by changing the configuration so that it is the second component, rather than the first, that is brought by a configuration change $\phi_2$ to the same local state at $X$ and $X'$. Clearly, if the composite points $X$ and $X'$ happen to be materially isomorphic, both definitions coincide. If $s$ and $s^*$ are material isomorphisms between the first and second components, respectively, of $X$ and $X'$ , and if the pairs $X,Y$ and $X',Y'$ are 1-compatible, then a necessary and sufficient condition for  these pairs $X, Y$ and $X', Y'$ to be also 2-compatible is that $n={ s^*}^{-1} s$ must belong to the normalizer of $m$ in the general linear group.\footnote{For a single element $m$ of a group, the normalizer coincides with the centralizer, namely, the subgroup consisting of all elements $n$ that commute with $m$.} Notice that $n$ is precisely the misalignment of $X'$ with respect to $X$.

The uniformity of a binary composite requires more than the simultaneous 1- and 2- compatibility of all the commuting squares in the coarse groupoid. We may say that the composite is uniform if its core groupoid is transitive, as shown in \cite{symmetry}. For the case of two triclinic components this statement is self evident, since the core groupoid of a material double groupoid consists of all pairs of identical arrows, while all other commutative squares are eliminated from the picture. We note also that if the core is not transitive, the only information that we can gather from it is that the composite is not uniform, but we cannot quantify a local measure of non-uniformity. Clearly, as demonstrated in Section \ref{sec:measures}, these analyses can be carried out without recourse to or recognition of the underlying structure of a double groupoid, just as is the case with many theories when it comes to their practical applications.

\subsection{The commutation condition}

We have defined the material double groupoid $\mathcal Q$ as the collection of all commutative squares in the coarse double groupoid $\square(\mathcal Z,{\hat{ \mathcal Z}})$. We will presently interpret the commutation condition in terms of the coupled symmetries just introduced. A square in the double groupoid, such as the one represented on the left in Figure \ref{fig:complementary}, consists of two horizontal and two vertical arrows. Since we have assumed the two side groupoids to be transitive (that is, the two components are materially uniform), each such square ${\mathbb A} \in {\mathcal Q}$ induces a {\it complementary} square ${\mathbb A}^*$ in $\square(\hat{\mathcal Z}, {\mathcal Z})=\square^T({\mathcal Z}, \hat{\mathcal Z})$, the transposed of the coarse groupoid. This square, represented on the right side of Figure \ref{fig:complementary}, is unique, since we have also assumed the two components to be triclinic. We have denoted the arrows in the complementary square with an asterisk. Thus, for instance $t^*$ represents a material isomorphism in the second component of the composite, while ${\hat t}^*$ represents a material isomorphism in the first.
\begin{figure}[H]
\begin{center}
\begin{tikzpicture}[baseline=(current  bounding  box.center)]
  \draw[thick,-stealth'] (3,0) -- (0,0);
  \draw[thick,-stealth'] (3,3) -- (0,3);
  \draw[thick,-stealth'] (3,0) -- (3,3);
  \draw[thick,-stealth'] (0,0) -- (0,3);
\node at (1.5,1.5) {$\mathbb A$};
\node[right] at (3,1.5) {${\hat s}$};
\node[left] at (0,1.5) {${\hat t}$};
\node[above] at (1.5,3) {$t$};
\node[below] at (1.5,0) {$s$};
\node[below right] at ((3,0) {$W$};
\node[above right] at ((3,3) {$X$};
\node[above left] at ((0,3) {$Z$};
\node[below left] at ((0,0) {$Y$};
\begin{scope} [xshift=6cm]
  \draw[thick,-stealth'] (3,0) -- (0,0);
  \draw[thick,-stealth'] (3,3) -- (0,3);
  \draw[thick,-stealth'] (3,0) -- (3,3);
  \draw[thick,-stealth'] (0,0) -- (0,3);
\node at (1.5,1.5) {${\mathbb A}^*$};
\node[right] at (3,1.5) {${\hat s}^*$};
\node[left] at (0,1.5) {${\hat t}^*$};
\node[above] at (1.5,3) {$t^*$};
\node[below] at (1.5,0) {$s^*$};
\node[below right] at ((3,0) {$W$};
\node[above right] at ((3,3) {$X$};
\node[above left] at ((0,3) {$Z$};
\node[below left] at ((0,0) {$Y$};
\end{scope}
\end{tikzpicture}
\end{center}
\caption{A square in $\mathcal Q$ and its complementary square in $\square(\hat{\mathcal Z}, {\mathcal Z})$ }
\label{fig:complementary}
\end{figure}
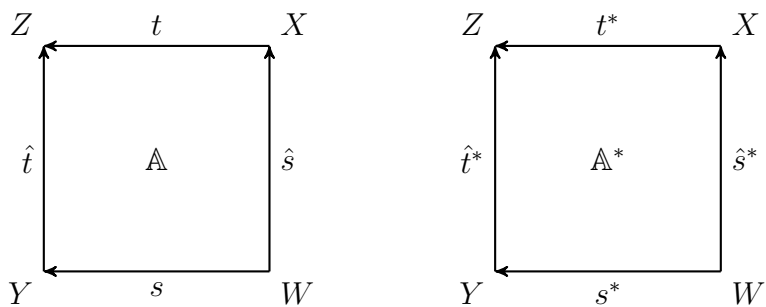

The complementary square ${\mathbb A}^*$ is not necessarily commutative and, hence, it does not necessarily belong to the transposed material groupoid ${\mathcal Q}^T$. We only have the following equations
\begin{enumerate}
\item\label{cond1} $t {\hat s} ={\hat t} s$, by the commutativity of squares in $\mathcal Q$,
\item\label{cond2} $t {\hat s}^*={\hat t}^* s$, by uniformity of the triclinic component 1,
\item\label{cond3} $t^* {\hat s}={\hat t}  s^*$, by uniformity of the triclinic component 2.
\end{enumerate}

Consider the pairs $W,Y$ and $X,Z$.  Their respective misalignments are $m=s^{*^{-1}} s$ and $m'=t^{*^{-1}} t$. Moreover, $\hat s$ is a material isomorphism of component 2 of the composite. Therefore,
\begin{equation}
{\hat s} m {\hat s}^{-1}={\hat s} s^{*^{-1}} s  {\hat s}^{-1} ={\hat s} s^{*^{-1}} {\hat t}^{-1} t={\hat s} {\hat s}^{-1} t^{*^{-1}} t=m',
\end{equation}
where Conditions \ref{cond1} and \ref{cond3} above have been used. We conclude, therefore, that these two pairs are 2-compatible. Following the same procedure, one can prove that pairs $W,X$ and $Y,Z$ are 1-compatible. In conclusion, the meaning of the commutation condition (which selects which squares enter in the material double groupoid $\mathcal Q$) is the 1- and 2-compatibility of the respective pairs of corners in each square.

The following proposition is a direct result of the above definitions and properties.

\begin{prop} The collection of complemetary squares of $\mathcal Q$ coincides with  ${\mathcal Q}^T$ if, and only if, the opposite pairs of all squares of $\mathcal Q$ are simultaneously 1- and 2-compatible.
\end{prop}

\subsection{Example}             

The commutation condition is clearly invariant under a chage of reference configuration. To better grasp its meaning, therefore, it is possible always to change the configuration in such a way that all 4 points are in exactly the same state as far as the first component of the mixture is concerned. In that case, the arrows $s,t, {\hat s}^*$ and ${\hat t}^*$ in Figure \ref{fig:complementary} are represented by the unit matrix $[I]$. In other words, the corresponding material isomorphisms are Euclidean translations.\footnote{Notice that an arrow represented by the unit matrix is not necessarily a unit (loop-shaped) arrow, since the two end points of the arrow may be different.} Enforcing the commutation condition, we obtain the general result that, given a square $\mathbb A$ in the material double groupoid ${\mathcal Q}$ of the composite, a change of configuration can be produced such that $\mathbb A$ attains the matrix representation of its 4 arrows given by

\begin{equation}
\begin{tikzpicture}[baseline=(current  bounding  box.center)]
  \draw[thick,-stealth'] (3,0) -- (0,0);
  \draw[thick,-stealth'] (3,3) -- (0,3);
  \draw[thick,-stealth'] (3,0) -- (3,3);
  \draw[thick,-stealth'] (0,0) -- (0,3);
\node at (1.5,1.5) {$\mathbb A$};
\node[right] at (3,1.5) {$[\hat S]$};
\node[left] at (0,1.5) {$[\hat S]$};
\node[above] at (1.5,3) {$[I]$};
\node[below] at (1.5,0) {$[I]$};
\node[below right] at ((3,0) {$W$};
\node[above right] at ((3,3) {$X$};
\node[above left] at ((0,3) {$Z$};
\node[below left] at ((0,0) {$Y$};

\end{tikzpicture}
\end{equation}
This diagram reveals the intuitive meaning of 1-compatibility between the pairs $W,X$ and $Y,Z$. The complementary square ${\mathbb A}^* \in \square(\hat{\mathcal Z}, {\mathcal Z})$ has the form
\begin{equation}
\begin{tikzpicture}[baseline=(current  bounding  box.center)]
  \draw[thick,-stealth'] (3,0) -- (0,0);
  \draw[thick,-stealth'] (3,3) -- (0,3);
  \draw[thick,-stealth'] (3,0) -- (3,3);
  \draw[thick,-stealth'] (0,0) -- (0,3);
\node at (1.5,1.5) {${\mathbb A}^*$};
\node[right] at (3,1.5) {$[I]$};
\node[left] at (0,1.5) {$[I]$};
\node[above] at (1.5,3) {$[P^*]$};
\node[below] at (1.5,0) {$[Q^*]$};
\node[below right] at ((3,0) {$W$};
\node[above right] at ((3,3) {$X$};
\node[above left] at ((0,3) {$Z$};
\node[below left] at ((0,0) {$Y$};

\end{tikzpicture}
\end{equation} 
Since the points $W$ and $Z$ are materially isomorphic as points of the second component, it follows that
\begin{equation}\label{eq200}
[P^*][{\hat S}] = [{\hat S}] [Q^*].
\end{equation}
Thus, the complementary square is not necessarily commutative. It is so if, and only if, $[P^*]$ commutes with $[{\hat S}]$, which implies that $[P^*]=[Q^*]$. This in turn implies that the 1-compatible pairs $W,X$ and $Y,Z$ would also be 2-compatible.

As a concrete example, consider 4 points lying on the Cartesian $x,y$ plane, such that component 1 is in the same state at all 4 points, while component 2 is at states that differ by rotations about the $z$-axis. Since these rotations are commutative, condition (\ref{eq200}) implies that the pairs are both 1- and 2-compatible. In contradistinction, if the rotations are about arbitrary axes, this property is lost in general. Even in the case where the opposite pairs are both 1- and 2- compatible, no two corners of the square are necessarily materially isomorphic as far as the composite is concerned. Consider, for example, the case of rotations of component 2 about $z$ of 10$^o$ between $W$ and $X$, and of 30$^o$ between $W$ and $Y$. By commutativity of the square, the rotation between $Y$ and $Z$ must be of 10$^o$, and consequently of 30$^o$ between $X$ and $Z$. Thus the rotations of component 2 at $X$, $Y$, and $Z$ are, respectively, $10$, $30$, and $40$ degrees, namely, all different. It is the amounts of misalignment of opposite pairs which are equal.

\section{A pictorial glimpse of the double Lie algebroid and the characteristic distribution}
\label{sec:algebroid}

The Lie algebra of a Lie group can be intuitively regarded as the description of the behaviour of the group in a small vicinity of the unit element. The geometric entity underlying the Lie algebra is the tangent space to the group manifold at the group unit. In the case of a Lie groupoid $\mathcal Z\rightrightarrows {\mathcal B}$, a similar construction leads to the notion of the associated Lie algebroid. At each point of $X \in \mathcal B$ there is a well-defined unit $\epsilon (X) \in {\mathcal Z}$, that is, a loop-shaped arrow whose source and target are $X$ and whose composition on the left with any arrow with target $X$ leaves this arrow invariant.

The collection of all arrows arriving at all $X \in {\mathcal B}$ can be conceived as having the structure of a fiber bundle over $\mathcal B$ with projection map $\beta$, the target map of the original groupoid. The unit at $X$ clearly belongs to the fiber of this bundle over $X$. The geometrical entity underlying the Lie algebroid of a Lie groupoid is the vector bundle obtained as the union of the tangent spaces to these fibers at the corresponding units. Given a vector in the tangent space at $\epsilon (X)$, we can project it on the tangent space $T_X{\mathcal B}$ using the derivative of the source map, $\alpha:{\mathcal Z} \to {\mathcal B}$, of the original groupoid.

Following \cite{mms}, these rather abstract notions can be visualized as depicted in Figure \ref{fig:algebroid}. The identity element at $X$ is represented as the solid loop-shaped arrow. A small neighbourhood of this arrow can be explored in any combination of two modalities. In the first, we can move to other nearby loop-shaped arrows, namely, to other elements of the vertex group at $X$. In this way, we obtain the Lie algebra of this vertex group of local symmetries. The second modality consists of exploring the arrows that start at a point in the vicinity of $X$ and end at $X$. Since the opening of each arrow is small, this gap can be regarded as a non-zero vector tangent to $\mathcal B$ at $X$, while the loop-shaped arrows of the first modality project onto the zero vector.
\begin{figure}[H]
\begin{center}
\begin{tikzpicture} [scale=1.5]

\begin{scope}[shift={(2.5,-1.2)}, scale=2]
\draw[fill=gray!40] (0,0) to [out=20, in=160] (2,0) to [out=80, in=225] (2.5,1)  to [out=160, in=20] (0.5,1) to [out=225, in=80] (0,0) ;

\node at (1.9,0.2) {$\mathcal B$};
\end{scope}

\begin{scope}[shift={(5,0)}]
\draw[fill](0,0) circle[radius=0.05cm];
\node[below]  at (0,0) {$_X$};
\draw[fill](0.4,0) circle[radius=0.05cm];
\node[below]  at (0.80,0) {$_{X+dX}$};
\draw [thick, -stealth'](0.1,0) to [out=50, in=0] (0,2) to [out=180, in=130] (-0.1,0);
\draw [thick, dashed, -stealth'](0.1,0) to [out=55, in=0] (0,1.7) to [out=180, in=125] (-0.1,0);
\draw [thick, dashed, -stealth'](0.1,0) to [out=60, in=0] (0,1.5) to [out=180, in=120] (-0.1,0);
\draw [thick, dashed,  -stealth'](0.4,0) to [out=45, in=0] (0,2.5) to [out=180, in=135] (-0.1,0);


\end{scope}

\end{tikzpicture}
\end{center}
\caption{Lie algebroid elements}
\label{fig:algebroid}
\end{figure}
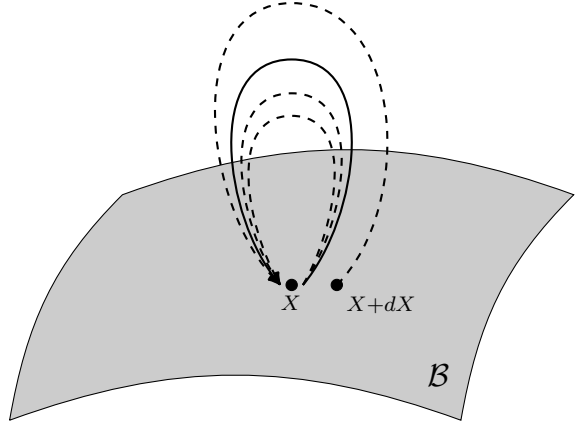

In the spirit of regarding the Lie algebroid as the collection of arrows in a small neighbourhood of each identity $\epsilon(X)$, we conclude that, if the projection of all these arrows covers the whole tangent space $T_X{\mathcal B}$, the situation is such that all points in the neighbourhood are connected by at least one arrow. If this is the case at all points $X$, therefore, we can say that we are in the presence of a transitive Lie algebroid. In a general case, however, the image of the projection into each tangent space may be a linear subspace of $T_X{\mathcal B}$, and the dimension of this subspace may vary from point to point. In this case, one obtains a {\it singular distribution}. If this distribution is integrable, the base manifold can be partitioned into transitive submanifolds of various dimensions, as discussed in \cite{mms1, jgp1, book}.

\begin{example} \label{ex:algebroid} {\rm{\bf{The pair Lie algebroid}:} The underlying space of the Lie algebroid of the pair groupoid of a manifold $\mathcal B$ is the tangent bundle $T{\mathcal B}$. Naturally, the Lie algebroid has also a well-defined bracket operation between sections of this bundle (that is, vector fields), but we are only looking at the underlying geometric object.}
\end{example}

The transition from a groupoid to a double groupoid and their corresponding Lie algebroids is technically more subtle. The underlying geometric object of a double Lie algebroid is a double vector bundle equipped with two Lie algebroid structures, as demonstrated in two classical articles by Mackenzie \cite{mac1, mac2}. Drastically simplifying the treatment, the procedure is carried out in two steps. The first step consists of constructing the Lie algebroid associated with, say, the vertical structure of the double groupoid $\mathcal Q$. The set of objects for this vertical structure is the horizontal side groupoid $\mathcal H$. Each ``point'' of this set is, therefore, a horizontal arrow $u$ and we want to explore, so to speak, a neighbourhood of the unit thereat. But the unit of the vertical structure at $u$ consists of the two horizontal sides $u$ and two loop-shaped vertical arrows. namely, the unit arrows ${\hat \epsilon}(\alpha(u))$ and ${\hat \epsilon}(\beta(u))$, as shown in the diagram (\ref{dg4a}).

As before, we can explore the neighbourhood of this unit by a combination of two modalities, namely the nearby loop-shaped arrows and the slightly open arrows near the units. To capture the essence of this construction, let us represent the second modality only, as done in Figure \ref{fig:firststep}, having in mind an example such as the pair double groupoid. We observe that the tails of the arrows arriving at $X$ and at $Z$ run over the respective tangent spaces $T_X{\mathcal B}$ and $T_Z{\mathcal B}$, thus spanning all nearby horizontal arrows $u+du$ with tail at $X+dX$ and tip at $Z+dZ$. Repeating this process for every arrow $u$ of the horizontal structure, we may say that what we have obtained is a groupoid (of horizontal arrows) whose set of objects is the vector bundle $T{\mathcal B}$.
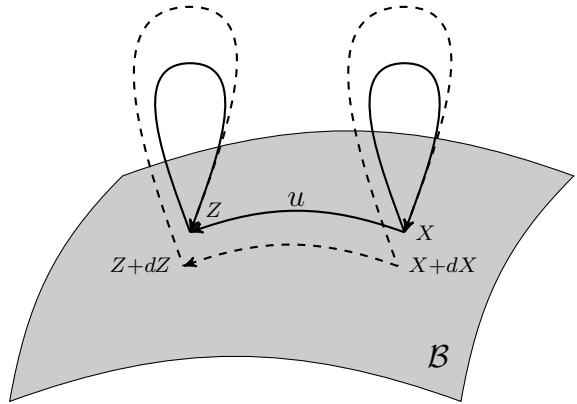
\begin{figure}[H]
\begin{center}
\begin{tikzpicture} [scale=3]

\draw[fill=gray!40] (0,0) to [out=20, in=160] (2,0) to [out=80, in=225] (2.5,1)  to [out=160, in=20] (0.5,1) to [out=225, in=80] (0,0) ;
\draw[thick, -stealth'] (1.75,0.75) to [bend right=20] (0.8,0.75);
\draw [thick, stealth'-](1.75,0.75) to [out=70, in=0] (1.75,1.5) to [out=180, in=110] (1.750,0.750);
\draw [thick, dashed, stealth'-](1.75,0.75) to [out=70, in=0] (1.75,1.75) to [out=180, in=110] (1.720,0.60);
\draw [thick, stealth'-](0.8,0.75) to [out=70, in=0] (0.8,1.5) to [out=180, in=110] (0.8,0.750);
\draw [thick, dashed, stealth'-](0.8,0.75) to [out=70, in=0] (0.8,1.75) to [out=180, in=110] (0.770,0.60);
\node[right] at (1.75,0.75) {$_X$};
\node[right] at (1.72,0.6) {$_{X+ dX}$};
\node[above right] at (0.82,0.77) {$_Z$};
\node[left] at (0.77,0.6) {$_{Z+dZ}$};

\node[above] at (1.27,0.82){$u$};

\node at (1.9,0.2) {$\mathcal B$};
\draw[thick,dashed, -stealth'] (1.72,0.6) to [bend right=20] (0.77,0.6);

\end{tikzpicture}
\end{center}
\caption{The first step in the construction of the double Lie algebroid delivers a groupoid over the tangent bundle $T{\mathcal B}$}
\label{fig:firststep}
\end{figure}

This new groupoid with base $T{\mathcal B}$ can be used to start the second step of the construction of the double Lie algebroid. Indeed, we can attempt to construct the Lie algebroid of this new groupoid. What that would entail is bringing $X$ towards $Z$ thus making $u$ equal to the horizontal unit at $Z$ and exploring the neighbourhood while the tail of the arriving arrows moves on the tangent space $T_Z{\mathcal B}$. We may say, without the due degree of precision, that at each point $Z$ of $\mathcal B$ we are considering the {\it double identity} of $\mathcal Q$, which is a square whose four corners coincide with $W$ and whose arrows are all identities, two from the horizontal and two from the vertical groupoid. In Figure \ref{fig:doublealgebroid} we represent this square as a clover leaf. A neighbourhood of this element lives in the iterated tangent bundle $T(T{\mathcal B})$.

The construction sketched above could have been reversed, in the sense that one could have first concentrated the attention on the horizontal (rather than the vertical) structure of $\mathcal Q$, thus obtaining a different groupoid with $T{\mathcal B}$ as its set of objects. Can perhaps the two final results for the double Lie algebroid be different? The resolution of this apparent contradiction is given by the fact that the iterated tangent bundle possesses two different bundle structures, which turn out to be equivalent via the so-called canonical involution (or switch operation).

Even after the intuitive presentation just valiantly attempted, it would appear that the technicalities are so overwhelming that no specific results pertaining to the application at hand could be obtained. Neverhteless, some admittedly non-rigorous but quite convincing results can be glimpsed. From the general picture elicited for the construction of the double Lie algebroid of the material double groupoid $\mathcal Q$, we have learned that, in the final analysis, we need to investigate all the commutative squares that may exist in the vicinity of the double units at each point of $\mathcal B$.  In Figure \ref{fig:doublealgebroid}, we have suggestively represented the double unit at a point $W$ as a clover leaf, for obvious reasons. In the same figure, we have drawn a nearby generic square. The result of this investigation should at the very least confirm and justify the measures of non-uniformity proposed in Section \ref{sec:measures}.

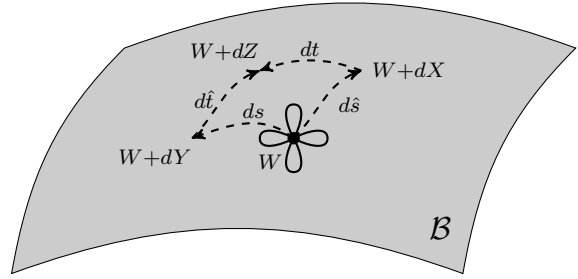
\begin{figure}[H]
\begin{center}
\begin{tikzpicture} [scale=1.5]

\begin{scope}[shift={(2.5,-1.2)}, scale=2]
\draw[fill=gray!40] (0,0) to [out=20, in=160] (2,0) to [out=80, in=225] (2.5,1)  to [out=160, in=20] (0.5,1) to [out=225, in=80] (0,0) ;

\node at (1.9,0.2) {$\mathcal B$};
\end{scope}

\begin{scope}[shift={(5,0)}]
\draw[fill](0,0) circle[radius=0.05cm];
\node[below]  at (-0.20,-0.05) {$_W$};

\draw[thick] (0,0) to [out=120, in=180] (0,0.3) to [out=0, in=60] (0,0) to [out=30, in=90] (0.3,0) to [out=-90, in  = -30] (0,0) to [out=-60, in=0] (0,-0.3) to [out=180, in=-120] (0,0) to [out=-150, in=-90] (-0.3,0) to [out=90, in=150] (0,0);
\draw[thick, dashed, -stealth'] (0,0) to [out=45, in=200] (0.6,0.6);
\draw[thick, dashed, -stealth'] (0.6,0.6) to [out=160, in=20] (-0.3,0.6);
\draw[thick, dashed, -stealth'] (-0.9,0) to [out=45, in=200] (-0.3,0.6);
\draw[thick, dashed, -stealth'] (0,0) to [out=140, in=20] (-0.9,0);
\node[right] at (0.6,0.6) {$_{W+dX}$};
\node[above left] at (-0.2,0.6) {$_{W+dZ}$};
\node[below left] at (-0.8,0.0) {$_{W+dY}$};
\node[right] at (0.3,0.3) {$_{d{\hat s}}$};
\node[above] at (0.15,0.63) {$_{dt}$};
\node[left] at (-0.6,0.35) {$_{d{\hat t}}$};
\node[above] at (-0.35,0.08) {$_{ds}$};

\end{scope}

\end{tikzpicture}
\end{center}
\caption{Double Lie algebroid square elements}
\label{fig:doublealgebroid}
\end{figure}

Consider the case of two smoothly uniform components of the triclinic kind (or with any discrete symmetry groups). Their material groupoids are, therefore, transitive Lie groupoids. Since the symmetry groups are trivial, the (smoothly varying) arrows in a neighbourhood are uniquely defined between each pair of points. Let the arrow $H$ (material isomorphism) between two points $X$ and $X'$ of one material be expressed in terms of the implant maps $P(X)$  and $P(X')$ of an archetype $A$, as suggested in Figure \ref{fig:arrow}, namely,
\begin{equation} \label{eq100}
H^I_J=P^I_\alpha(X') P^{-\alpha}_J (X),
\end{equation}
The differential of $H$ under a small change of the coordinates of the end points is
\begin{equation} \label{eq101}
dH^I_J=P^I_\alpha(X') P^{-\alpha}_{J,K}(X) dX^K +  P^I_{\alpha,K}(X') P^{-\alpha}_J dX'^K,
\end{equation}
where commas indicate partial derivatives. Recalling the definition of the Christoffel symbols of the material connection
\begin{equation} \label{eq102}
\Gamma^I_{JK}=-P^{-\alpha}_J P^I_{\alpha,K},
\end{equation}
Equation (\ref{eq101}) can be written as
\begin{equation}  \label{eq103}
dH^I_J=H^I_M \Gamma^M_{JK} (X) dX^K  -  H^M_J \Gamma^I_{MK}(X') dX'^K.
\end{equation}
For the unit arrow, we have that $X'=X$ and $H=I$. Therefore, in a neighbourhood of a unit, we obtain
\begin{equation}  \label{eq103a}
dH^I_J= \Gamma^I_{JK} (X) (dX^K  - dX'^K).
\end{equation}
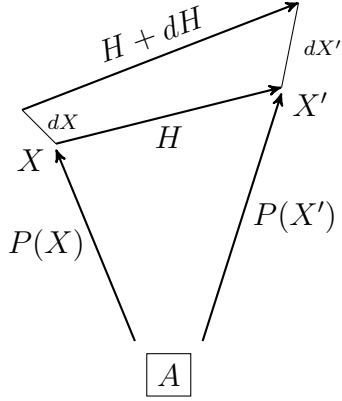
\begin{figure}[H]
\begin{center}
\begin{tikzpicture} [scale=1.5]

\node[draw] (r) at (0,0) {$A$};
\draw[thick,-stealth'] (-0.3,0.3)--(-1,2);
\draw[thick,-stealth'] (0.3,0.3)--(1.,2.5);
\draw[thick,-stealth'] (-1,2.05)--(1,2.55);
\draw (-1,2.05)--(-1.3,2.35);
\draw (1,2.55)--(1.15,3.3);
\draw[thick,-stealth'] (-1.3,2.35)--(1.15,3.3);
\node[left] at (-1,1.9) {$X$};
\node[right] at (1,2.4) {$X'$};
\node[below] at (0,2.3) {$H$};
\node[above, rotate=20] at (-0.08,2.82) {$H+dH$};
\node[right] at (-1.17,2.25) {$_{dX}$};
\node[right] at (1.1,2.9) {$_{dX'}$};
\node[right] at (0.65,1.4) {$P(X')$};
\node[left] at (-0.65,1.15) {$P(X)$};
\end{tikzpicture}
\end{center}
\caption{Small change of an arrow by varying its ends}
\label{fig:arrow}
\end{figure}

This result can be used to implement the commutation condition (\ref{double1}) in a square near the double unit, such as the one shown in Figure \ref{fig:doublealgebroid}. The infinitesimal version of Equation (\ref{double1}) is
\begin{equation}  \label{eq104}
d{\hat s} + dt=ds + d{\hat t}.
\end{equation}
In terms of the matrix expressions of the arrows involved, this equation can be written as
\begin{equation}  \label{eq105}
d{\hat S}^I_J  + dT^I_J  =  dS^I_J + d{\hat T}^I_J.
\end{equation}
Using Equation (\ref{eq103a}) for each of these arrows yields
\begin{equation} \label{eq106}
{\hat\Gamma}^I_{JK}(-dX^K)+\Gamma^I_{JK} (dX^K-dZ^K) = \Gamma^I_{JK} (-dY^K) + {\hat\Gamma}^I_{JK} (dY^K-dZ^K).
\end{equation}
Recognizing the non-uniformity tensor $\bf B$ introduced in Equation (\ref{eq1}), we can write this commutation condition around the double unit as
\begin{equation} \label{eq107}
B^I_{JK} (dX^K+dY^K-dZ^K) = 0.
\end{equation}
If the tensor $\bf B$ vanishes identically, we  arrive at the conclusion that all squares are commutative, which in turn implies that the composite is uniform. At the other extreme, if the tensor $\bf B$ is not annihilated by any vector $v^K$, we obtain that the only surviving commutative squares must satisfy the condition $dZ^K=dX^K+dY^K$. In particular, if points $X$ and $Y$ are made to coincide with $W$, then $dX^K=dY^K=0$, and we conclude that the only commutative square that survives is the double unit itself. In other words, there are no non-trivial squares in the core, such that the arrows of the two side groupoids between $W$ and $Z$ coincide. This is the case of a totally intransitive groupoid. On the other hand, if at each point the annihilator of $\bf B$ is a linear subspace of the tangent space, we obtain a singular distribution on $\mathcal B$. In this case, the body can be regarded as a union of uniform submanifolds of various dimensions.

If the symmetry group of one or both of the components is a Lie group, a similar argument will lead to the appearance of an added element of its Lie algebra. Thus, in the case of an isotropic solid, an arbitrary skew symmetric term appears on the right-hand side of Equation (\ref{eq107})   which results in the emergence of a different measure of non-uniformity, corresponding exactly to the differential version of that given in Equation (\ref{eq3}) involving the difference between the Riemannian metrics of the two materials of the composite.

\section{Summary}

A combination of two uniform material structures coexisting upon the same substrate gives rise in general to a non-uniform mixture. This lack of uniformity may be viewed as a continuous distribution of material defects and, hence, as a possible driving force for material evolution. Accordingly, the definition of a quantitative measure of defectivity is necessary to formulate evolution equations. In this paper, we have proposed such measures for various combinations of symmetry types. These definitions appear to be novel. 

Although theaforementioned aspects of the theory of binary composites can be treated by conventional means, we also attempted to reveal the underlying geometric object behind the theory of uniformity of binary composites. In the spirit of the well-established tradition that advocates the use of differential geometric tools in this field, we have argued that the appropriate framework for binary composites is the theory of double groupoids and their associated double Lie algebroids and characteristic distributions. The proposed measures of non-uniformity are shown to arise naturally within this framework.

\end{document}